\documentclass[12pt,twoside]{article}

\usepackage{float}
\usepackage{amsfonts}
\usepackage{amssymb}
\usepackage{amsmath}
\usepackage{amsthm}
\usepackage{mathrsfs}		
\usepackage{stmaryrd}
\usepackage{listings}
\usepackage{multicol}
\usepackage{amsmath}
\usepackage[english]{babel}
\usepackage{graphicx}
\usepackage{color}
\usepackage{boxedminipage}
\usepackage{url}

\usepackage{tikz}
\usepackage{graphicx}
\usepackage{pgf}
\usepackage{tikz}
\usetikzlibrary{arrows,automata}
\usepackage[latin1]{inputenc}
\tikzstyle{every picture}+=[remember picture]

\usepackage{datetime}
\newdate{date}{10}{07}{2014}

\usepackage[ruled,vlined,commentsnumbered]{algorithm2e}
\usepackage{pgf}
\usepackage{tikz}
\usetikzlibrary{matrix}
\usetikzlibrary{arrows,automata}
\usetikzlibrary{arrows,shapes}
\usepackage{float}
\usepackage{caption}
\usepackage{subcaption}
\usetikzlibrary{arrows,matrix,positioning}

\lstdefinestyle{customc}{
  belowcaptionskip=1\baselineskip,
  breaklines=true,
  frame=L,
  xleftmargin=\parindent,
  language=C,
  showstringspaces=false,
  basicstyle=\footnotesize\ttfamily,
  keywordstyle=\bfseries\color{green!40!black},
  commentstyle=\itshape\color{purple!40!black},
  identifierstyle=\color{blue},
  stringstyle=\color{orange},
}

\newcommand\fixmet[1]{}
\newcommand{\N}{\ensuremath{\mathbb{N}}}
\DeclareMathOperator{\A}{\textsf{A}}

\newtheorem{ass}{Assumption}[section]
\newtheorem{fact}{Fact}[section]
\newtheorem{prop}{Proposition}[section]

\newtheorem{LM}{Lemma}[section]

\newtheorem{thm}{Theorem}[section]
\newtheorem{df}{Definition}[section]
\newtheorem{cor}{Corollary}[section]

\theoremstyle{pourlesremarques}

\newtheorem{ex}{Example}[section]

\newcommand{\R}{\mathbb{R}}
\renewcommand{\l}{\lambda}
\newcommand{\C}{\mathbb{C}}
\newcommand{\Q}{\mathbb{Q}}
\renewcommand{\H}{\mathbb{H}}

\newcommand{\G}{\mathbb{G}}

\newcommand{\Z}{\mathbb{Z}}

\renewcommand{\i}{\mathbf{id_E}}
\newcommand{\f}{\mathbf{f}}
\newcommand{\bb}{\mathbf{b}}
\newcommand{\cb}{\mathbf{c}}
\newcommand{\U}{\mathbb{U}}
\newcommand{\Ab}{\mathbf{A}}
\newcommand{\F}{\mathbf{F}}

\begin{document}
\title{On the Termination of Linear and Affine Programs over the Integers}

\author{Rachid Rebiha\thanks{Instituto  de Computa\c{c}\~{a}o, Universidade Estadual  de Campinas, 13081\-970  Campinas,  SP.  Pesquisa desenvolvida com suporte financeiro da FAPESP, processos 2011\/08947\-1 e FAPESP BEPE 2013\/04734\-9 } \and  Arnaldo Vieira Moura\thanks{Instituto  de Computa\c{c}\~{a}o, Universidade Estadual  de Campinas, 13081\-970  Campinas,  SP.} \and Nadir Matringe \thanks{ Universit\'{e} de Poitiers, Laboratoire Math\'{e}matiques et Applications and Institue de Mathematiques de Jussieu Universit\'{e} Paris 7-Denis Diderot, France.}}

\date{\displaydate{date}}
\maketitle

\begin{abstract} 
The termination problem for affine programs over the \emph{integers} was left open in~\cite{Braverman}. For more that a decade, it has been considered and cited as a challenging open problem.  To the best of our knowledge, we present here the most complete response to this issue: we show that 
termination for affine programs over $\Z$ is decidable under an assumption holding for almost all affine programs,
except for an extremely small class of zero Lesbegue measure. We use the notion of \emph{asymptotically non-terminating initial variable values} ($ANT$, for short) for linear loop programs over $\Z$. Those  values are directly associated to initial variable values for which the corresponding program does not terminate. 
We reduce the termination problem of linear affine programs over the integers to the emptiness check of a specific $ANT$ set of  initial variable values. 
For this class of linear or affine programs,  we prove that the corresponding $ANT$ set is a semi-linear space and we provide a powerful computational methods allowing the automatic generation of these $ANT$ sets.  
Moreover, we are able to address the conditional termination problem too.
In other words, by taking $ANT$ set complements, we obtain a precise under-approximation of the set of inputs for which the program does terminate.  
\end{abstract}
\section{Introduction}

The \emph{halting problem} is equivalent to the problem of deciding whether a given program will eventually terminate when
running with a given input. The \emph{termination problem} can be
stated as follows: given an arbitrary program, decide whether the
program eventually halts for every possible input. Both
problems are known to be undecidable, in general~\cite{turing}. 
The \emph{conditional termination problem} \cite{Cook:2008} asks for preconditions representing input data that will cause the program to terminate when run with such input data.  
In practice, this problem appears to be central for the verification of liveness properties that any well behaved and engineered system must guarantee. 
As it happens frequently, a program may terminate only for a specific set of input data values. 
Also, generating input data that demonstrates critical defects and vulnerabilities in programs  allows for new looks at these liveness properties.

In the present work we address the termination and conditional termination problem over linear and affine $\textsf{while}$ loop programs over the \emph{integers}.  In matrix terms, this class of programs can be expressed in the following form:
$$\texttt{while (F$\cdot$x > b) \{x := A$\cdot$x+c\}},$$
where $A$ and $F$ are matrices, $b$ and $c$ are
vectors over the integers, and $x$ is a vector of variables  over $\N$ or $\Z$. The loop condition is a conjunction of linear or affine inequalities and the assignments to each  of the variables in the loop instruction block are affine or linear forms. 
Static analysis and automated verification methods for programs presented in a more complex form can often be reduced to the study of a program expressed in this basic affine form~\cite{Cook:2006}. 
 
Recent approaches for termination analysis of imperative loop programs have focused on partial decision procedures based on the discovery and synthesis of ranking functions \cite{Colon:2001, Colon02}. 
Several interesting techniques are based on the generation of \emph{linear} ranking functions for linear loop programs~\cite{Bradley05linearranking, Bradley05terminationanalysis}. 
There are also effective heuristics \cite{Dams,Colon02} and complete methods for the synthesis of \emph{linear} ranking functions \cite{Podelski04}. On the other hand, there are simple linear terminating loop programs for which 
it can be proved that there are no linear ranking functions.

On the problem of decidability results for  termination of linear and affine programs, the work of Tiwari et al. \cite{tiwaricav04} is often cited when treating linear programs over the \emph{reals}. 
For linear programs over the \emph{rationals} and \emph{integers}, some of those theoretical results have been extended~\cite{Braverman}.  
 But the termination problem for \emph{general affine} programs over the 
 \emph{integers} was left open
in~\cite{Braverman}. For more than a decade, it has been considered and cited as a challenging open problem \cite{Braverman, tiwaricav04, JO, Ben-Amram:2012, Bozga}.  
This question was considered, but not completely answered, in \cite{tiwaricav04, Braverman}. 
Recently, in \cite{JO}, using strong results from analytic number theory and diophantine 
geometry, the authors were able to answer positively this question,
but only when the corresponding  transition matrices were restricted to a semi-simple form. 
In that work,  $ANT$ sets are not explicitly computed, because it is very 
hard by their approach --- in fact, it is not proved they form semi-linear spaces, --- and quantifier elimination is used.
To the best of our knowledge, we present here the most complete response to this open problem: we show that 
termination for general affine programs over $\Z$ is decidable under an assumption holding for almost all affine programs except for a very small class of zero Lesbegue measure.
Our contribution is not limited to such theoretical decidability results, 
as we also provide efficient computational methods to decide termination and conditional termination for this class of programs. 
More specifically, in our approach the computation of the $ANT$ set becomes simple, and it is described explicitly as a semi-linear space, without using quantifier elimination.\\

Concerning the termination and the conditional termination analysis problems,
 we could cite briefly the following recent developments.  
The framework presented in \cite{Cousot:2012} is devoted to approaches establishing termination by abstract interpretation over the termination semantics. 
The approach exposed in \cite{Cook:2008} searches for non-terminating program executions. 
The recent literature on \emph{conditional (non-)termination} narrows down to the works presented in
\cite{Gulwani:2008,Cook:2008, Bozga}. 
The methods proposed in \cite{Gulwani:2008} allow for the generation of non-linear preconditions.
In \cite{Cook:2008}, the authors derived termination preconditions for simple
programs ---~with only one loop condition~--- by guessing a ranking function
and inferring a supporting assertion. 
Also, the interesting approach
provided in \cite{Bozga} focuses mostly on proofs of decidability and consider several systems and models, but is restricted to 
two specific subclasses of linear relations. 
Despite tremendous progress over the years~\cite{Braverman, tiwaricav04,Bradley05Manna, Chen2012, Cousot:2012, Cook:2006, Ben,Gulwani:2008,Cook:2008, 2013:POPL}, 
the problem of finding a practical, sound and complete method, \emph{i.e.}, an encoding leading deterministically to an algorithm, for determining (conditional) termination remains very challenging.
Some more closely related works, \emph{e.g.} \cite{Gulwani:2008,Cook:2008, Bozga, JO, Braverman, tiwaricav04}, will be discussed in more details in Section \ref{discussion}.

Our initial investigations were reported in~\cite{TR-IC-13-07, TR-IC-13-08} where we discussed  termination analysis algorithms that ran in polynomial time complexity. Subsequent studies considered the set of \emph{asymptotically non-terminating initial variable values} ($ANT$, for short) whose elements are directly related to input values for which the loop does not terminate~\cite{scss2013rebiha}.
In that work, we approached the problem of generating the $ANT$ set for a restricted class of linear programs over the reals, with only one loop condition, and where the associated linear forms of the loop lead to diagonalizable systems with no complex eigenvalues.  
In \cite{TR-IC-14-09}, we showed how to compute the $ANT$ set for linear or affine programs over $\R$. 
In that work we also successfully treated the case of linear or affine programs over $\Z$ in cases where the transition matrices admit a real spectrum.
Here, we remove these restrictions. We show how to handle complex eigenvalues, linear affine programs over $\R$, $\Q$, $\N$ or $\Z$, with conjunctions of several loop conditions, and where the system does not have to be diagonalizable. 
We thus drastically generalize the earlier results in \cite{scss2013rebiha, TR-IC-14-09}.  
Further, we introduce new static analysis methods that compute $ANT$ sets, and also yield a set of initial inputs values for which the program does terminate.  This attests the innovation of our contributions, \emph{i.e.}, none of the other mentioned works is capable of generating such critical information for non-terminating loops.

We summarize our contributions as follows, with all 
results rigorously stated and proved:
\begin{description}
\item[\it On static input data analysis\rm:]\mbox{}
\begin{itemize}
\item We recall the important key concept of an $ANT$ set \cite{scss2013rebiha} for linear loop programs over the integers. 
Theorems \ref{equivalence}, \ref{gen2hom} and \ref{aff2gen} already show the importance of $ANT$ sets. These results provide us with \emph{necessary and sufficient conditions} for the termination of linear programs over the integers. 
\item 
For almost the whole class of linear, respectively affine, programs over $\Z$,
namely those with transition matrix 
satisfying our Assumption \ref{a1}, respectively Assumption \ref{a2}, 
we prove that the set of asymptotically non-terminating inputs can be computed explicitly as a \emph{semi-linear} space. 
Further, we show in Section \ref{measure} that almost all linear or affine programs belong to these classes, 
except for an extremely small specific class  of zero Lesbegue measure.
We are also capable of automatically generating a set of linear equalities and inequalities describing a semi-linear space that symbolically and exactly represents such $ANT$ sets. See Theorem \ref{generation-ANT}. 
\item Even if these results are mathematical in nature, they are easy to apply.
In a practical static analysis scenario, one only needs to focus on ready-to-use generic formulas that represent the $ANT$ sets for affine programs over $\Z$. See Definition \ref{eq1}, Eqs. (1), (2), (3) and (4). 
Such $ANT$ set representations allow for practical computational manipulations --- like union, intersection, complement and emptiness check, --- and practical implementations. 
\end{itemize}
\item[\it On static termination and conditional termination analysis\rm:]\mbox{} 
\begin{itemize}
\item We reduce the problem of termination for linear programs over the integers to the emptiness check of the corresponding $ANT$ set.   
This characterization of terminating linear programs provides us with a deterministic computational procedure to \emph{check program termination over $\Z$}, that is, we show that
an affine program $P$ is terminating if and only if $P$ has an empty $ANT$ set.  
\item  Also, the $ANT$ \emph{complement} set is a precise under-approximation of the set of terminating inputs for the same program.  
This complement set gives rise to a  loop precondition for termination. Thus, we obtain a computational methods for \emph{conditional termination} analysis.
\end{itemize}
\item[\it On decidability results for the termination problem over $\Z$\rm:]\mbox{}
\begin{itemize}
\item we obtain new decidability results for the program termination problem. 
Here, we successfully address the question left open in \cite{Braverman}, namely, we settle the decidability problem for program termination in the case of affine programs over the integers. Under our Assumption ($\mathcal{A}$) (see Fact~ \ref{a2}), we prove that the termination problem for affine programs over $\Z$ is decidable.  
\item We provide a complete measure analysis of this assumption and show that our decidability results holds almost for all linear/affine programs overs $\Z$ (i.e., we prove that the class of affine programs not satisfying our assumption is of Lesbegue measure zero.). 
\end{itemize}
\end{description}

Before concluding this section, we introduce a motivating example.
\begin{ex}\label{ex-motivation}\emph{(Motivating Example)} Consider the  program:

\begin{center}\begin{minipage}{0.5\textwidth}
\begin{lstlisting}
while(x-1/2y-2z>0){
 x:=-20x-9y+75z;
 y:=-7/20x+97/20y+21/4z;
 z:=35/97x+3/97y-40/97z;}
\end{lstlisting}
\end{minipage}\end{center}

\noindent The initial values of $x$, $y$ and $z$ are represented, respectively, by the parameters $u_1$, $u_2$ and $u_3$. Our prototype outputs the following $ANT$ set:
\noindent\begin{small}
\begin{verbatim}
Locus of ANT:
    [[u1<-u2+3*u3]]OR[[u1==-u2+3*u3,-u3<u2]]OR[[u1==4*u3,u2==-u3,0<u3]].
\end{verbatim}  
\end{small}

\begin{description}
\item[\small The static input data analysis\rm:] This semi-linear space  
represents symbolically all asymptotically initial values that are directly associated to initial values for which the program does not terminate.
\item[\small The conditional termination analysis\rm:] The complement of this set 
 is a precise under-approximation of the set of all initial values for which the program terminates.
\item[\small The termination analysis\rm:] The problem of termination is reduced to the emptiness check of this $ANT$ set.  
\end{description}
\end{ex}

The paper is organized as follows.  Section \ref{prelim} presents basic results from linear algebra and also defines the computational model used to represent linear loop programs in matrix notation.
Section \ref{ANT} introduces the notion of $ANT$ initial values, and 
presents the first important results for termination analysis. 
Section \ref{1loop} provides an efficient computational method for generating a symbolic representation of the $ANT$ set for linear homogeneous programs.  This section also states the ready-to-use formulas representing symbolically and exactly the $ANT$ sets for linear homogeneous programs.   Section \ref{general} reduces the study of generalized linear homogeneous and  affine loop programs
to that of linear homogeneous programs with a loop condition described by a single homogeneous inequality.  
Section \ref{measure} shows that our decidability result  holds for all linear  or affine programs, except for an extremely restricted class of zero measure programs.
We provide a complete discussion in Section \ref{discussion}. Finally, 
Section \ref{conclusion} states our conclusions. 
\section{Preliminaries}\label{prelim}
In Subsection \ref{algebra} we recall some classical concepts and results from linear algebra. 
In particular, we recall the Jacobian basis, in Theorem \ref{Jnr}, and note a very useful basis, in Theorem \ref{JU}. 
In subsection \ref{classification} we introduce  the computational model for loop programs in matrix notation, and
we provide an important classification for loop programs.   

\subsection{Linear Algebra}\label{algebra}
In the following, $E$ will be a finite dimensional vector space over $\R$. 
Let $A$ belong to $End_\R(E)$, the space of $\R$-linear maps from $E$ to itself, and 
let $E^*$ be the set of linear functionals in $E$, \emph{i.e.}, of mappings from $E$ to
$\R$. In the sequel we will assume that $f$ is a functional in $E^*$. 
We denote by $\mathcal{M}(p,q,\R)$ the space of $p\times q$ matrices.
When $p=q$ we may write $\mathcal{M}(p,\R)$. 
If $B$ is a basis of $E$, we denote by $Mat_B(A)\in \mathcal{M}(n,\R)$ the matrix
representation of $A$ in the basis $B$. 
Let $I_n$ be the identity matrix in $\mathcal{M}(n,\R)$, and let $\i$ the identity of $End_\R(E)$. 

We will denote by $Spec(A)$ the set of complex eigenvalues of $A$, by $Spec_\R(A)$ the set of real eigenvalues of $A$, and 
by $Spec_{>0}(A)$ the set of positive eigenvalues of $A$. In particular, we have 
$$Spec_{>0}(A)\subset Spec_\R(A)\subset Spec(A).$$ We will also denote by $|Spec(A)|$, the set $$\{|\mu|,\mu \in Spec(A)\},$$ and by $Spec_H(A)$, the intersection 
of $Spec(A)$ with the Poincaré upper half plane $$H=\{z\in \C, \ Im(z)>0\}.$$

For $\l \in Spec_\R(A)$, we write $E_{\lambda}$ for the characteristic subspace of $A$ 
associated to $\l$, which is the kernel $$Ker((A- \l I_d)^{d_\l}),$$ 
where $d_\l$ is the multiplicity 
of $\l$ in the characteristic polynomial $\chi_A$ of $A$. 
The non-real complex eigenvalues of $A$ 
come into couples of conjugate complex numbers.
If $\l$ is such an eigenvalue, with $\overline{\l}$ its conjugate, 
we write $E_{\{\l,\overline{\l}\}}$ for $$Ker[((A -\l I_d)\circ(A - \overline{\l} I_d))^{d_\l}],$$ where $\circ$ is the composition operator, and again $d_\l$ is the multiplicity 
of $\l$ in the characteristic polynomial of $A$. With these notations, we have the following direct sum decomposition:
$$E=\oplus_{\l \in Spec_\R(A)} E_\l \oplus _{\l \in Spec_H(A)} E_{\{\l,\overline{\l}\}}.$$ 

We recall the Jordan canonical basis theorem for $A$.
\begin{thm}\label{Jnr}
Let $\l$ belong to $Spec_\R(A)$. There is a basis $J_\l$ of $E_{\l}$ such that 
$Mat_{J_\l}(A_{|E_\l})=diag(U_{\l,1},\dots,U_{\l,r_\l})$ for a positive integer $r_\l$, where each $U_{\l,i}$ is of the form 
$\begin{pmatrix} \l & 1 &  &   &   &  \\
  & \l & 1&   &   &  \\
  &  & \ddots &\ddots &   & \\
  &  &  & \l & 1  & \\
  &  &  &  & \l  & 1 \\
    &  &  &  &   & \l
\end{pmatrix}$.
\end{thm}

For $\mu=a+ib=|\mu|e^{i\theta_\mu}$ a complex eigenvalue in $Spec_H(A)$, we denote by $s(\mu,\overline{\mu})$ the matrix 
$$s(\mu,\overline{\mu})=\begin{pmatrix} a & -b\\ b & a\end{pmatrix}= |\mu|r(\mu,\overline{\mu}),$$ where 
$$r(\mu,\overline{\mu})=\begin{pmatrix} cos(\theta_\mu) & -sin(\theta_\mu)\\ sin(\theta_\mu) & cos(\theta_\mu)\end{pmatrix}.$$ 
Similarly, we have the following theorem for $u$'s restriction to $E_{\mu,\overline{\mu}}$.

\begin{thm}\label{JU}
Let $\mu$ belong to $Spec_H(A)$. There is a basis $J_{\mu,\overline{\mu}}$ of $E_{\mu,\overline{\mu}}$ such that 
$$Mat_{J_{\mu,\overline{\mu}}}(A_{|E_{\mu,\overline{\mu}}})=diag(U_{{\mu,\overline{\mu}},1},\dots,U_{{\mu,\overline{\mu}},r_{\mu,\overline{\mu}}})$$ for a positive integer $r_{\mu,\overline{\mu}}$, where each $U_{{\mu,\overline{\mu}},i}$ is of the form 
$$\begin{pmatrix} s(\mu,\overline{\mu}) & I_2 &  &   &   &  \\
  & s(\mu,\overline{\mu}) & I_2&   &   &  \\
  &  & \ddots &\ddots &   & \\
  &  &  & s(\mu,\overline{\mu}) & I_2  & \\
  &  &  &  & s(\mu,\overline{\mu}) & I_2 \\
    &  &  &  &   & s(\mu,\overline{\mu})
\end{pmatrix}.$$
\end{thm}

\subsection{Classification of Loop Programs}\label{classification}

\noindent We recall, as it is standard in static program analysis, that a primed symbol $x'$ refers to the next  value of $x$ after a transition is taken. First, we present \emph{transition systems} as representations of imperative programs, and \emph{automata} as their computational models.

\begin{df}
A \emph{transition system} is given by $\langle x, L, \mathcal{T}, l_0,
\Theta \rangle$, where  $x=(x_1, ...,x_n)$ is a set of variables, 
 $L$ is a set of locations and $l_0\in L$ is the initial location. A \emph{state} is given by an interpretation of the variables in $x$. A \emph{transition} $\tau \in \mathcal{T}$ is given by a tuple
  $\langle l_{pre}, l_{post}, q_{\tau}, \rho_{\tau} \rangle$, where $l_{pre}$
  and $l_{post}$ designate the pre- and post- locations of $\tau$, respectively, and
  the transition relation $\rho_{\tau}$ is a first-order assertion over $x
  \cup x'$. The transition guard $q_{\tau}$ is a conjunction of inequalities over $x$.
  $\Theta$ is the initial condition, given as a first-order
  assertion over $x$.
The transition system is said to be \emph{linear} when
$\rho_{\tau}$ is an affine form, for all $\tau\in \mathcal{T}$. 
\end{df}
A loop program is a transition system with a single location and a single transition, written simply as $\langle x, l, \langle l, l, q_{\tau}, \rho_{\tau} \rangle,
l,\Theta \rangle$.

We will use the following matrix notation to represent loop programs and their transition systems. We also use simple and efficient procedures to captures the effects of sequential linear assignments into simultaneous updates.    

\begin{df} Let $P=\langle x, l, \langle l, l, q_{\tau}, \rho_{\tau} \rangle, l,\Theta \rangle$, with $x=(x_1,...,x_n)$,
be a loop program. We say that $P$ is a \emph{linear loop program} if:\\
\begin{itemize}
\item The transition guard is a conjunction of linear inequalities. We represent the loop condition in matrix form as $F x > b$ where $F\in\mathcal{M}(m,n,\R)$, and $b\in\R^m$.
By $F x>b$ we mean that each coordinate of  vector $F x$ is  greater than the corresponding coordinate of vector $b$.\\
\item The transition relation is a set of 
affine or linear forms. We represent the linear assignments in matrix form as $x:=Ax+c$, where $A\in\mathcal{M}(n,\R)$, and $c\in\R^n$.
\end{itemize}
The most general \emph{loop program} $P(A,F,b,c)$ is defined as $$\texttt{while (F$\cdot$x > b) \{x := A$\cdot$x+c\} }.$$
\end{df}    

We will use the following classification. 
\begin{df}\label{classprog}
From the more specific to the more general form:
\begin{description}
\item[\small Homogeneous\rm:]  We denote by $P^{\H}$ the set of programs of the form 
$$P(A,f):\texttt{while (f $\cdot$ x > b) \{x:=Ax\} },$$ where $f$ is a $1\times n$ row matrix corresponding to the loop condition, $b\in \R$, and $A\in \mathcal{M}(n,\R)$ corresponds to the list of assignments in the loop. 
\item[\small Generalized Homogeneous\rm:] We denote by $P^{\G}$ the set of programs of the form $$P(A,F):\texttt{while (F x > 0) \{x:=Ax\}},$$ where $F$ is a $(m\times n )$-matrix with rows corresponding to the $m$ loop conditions. We will sometimes write $P(A,F)=P(A,f_1,\dots,f_m)$, where the $f_i$'s are the rows of $F$.
\item[\small Affine\rm:]  We denote by $P^{\mathbb{A}}$ the set of programs of the form $$P(A,F,b,c):\texttt{while (F x > b) \{x:=Ax+c\} },$$ for $A$ and $F$ as above,  and $b, c\in \R^n$. 
\end{description}
\end{df}

\begin{ex}\label{motivation2}
Consider the homogeneous program of Example \ref{ex-motivation}. 
The sub-matrix $A=\begin{small}\begin{pmatrix}-20&-9&75\\7&8&-21\\-7&-3&26\end{pmatrix}\end{small}$ correspond to the simultaneous updates representing the sequential loop assignments and the vector $f=(1,-1/2,-2)^{\top}$ encodes the loop condition.
\end{ex}

In Section \ref{general}, we show that the termination analysis for the general class $P^{\mathbb{A}}$ can be reduced to  termination for programs in $P^{\H}$.
\section{The NT and ANT Sets}\label{ANT}

We present the new notion of \emph{asymptotically non-terminating} ($ANT$) values of a loop program \cite{scss2013rebiha}. It will be central in the analysis of non-termination.
We start with the definition of the $ANT$ set and then give the first important result for homogeneous linear programs. 
We will extend these results in Section \ref{general} to generalized linear homogeneous programs.
Then, problem of termination analysis for the general class of linear programs will be reduced to the generation and  the emptiness check of the $ANT$ set for homogeneous linear programs.

Let $E$, $A\in End_\R(E)$ and $f\in E^*$ be as introduced in Section \ref{algebra}.  In this section, we focus first on homogeneous programs  
$$P(A,f): \{\texttt{while (f$\cdot$x>0) x:=A x}\}.$$ 
Given a basis $B$ of $E$ we write  $\texttt{A}=Mat_B(A)$, $\texttt{f}=Mat_B(f)$, $\texttt{x}=Mat_B(x)$, and so on. From now on, we give definitions and statements in terms of 
programs involving linear maps, and let the reader infer the obvious adaptation for programs involving matrices.
We start by giving the definition of the termination and non-termination for this class of programs. 

\begin{df}\label{ter}
Let $P(A,f)\in P^{\H}$ and let $x\in E$ be an input for $P(A,f)$.
We say that $P(A,f)$ terminates on  $x$ if and only if there exists some $k \geq 0$ such that $f(A^{k}(x))\leq 0$; otherwise we say that $P(A,f)$ is non-terminating on $x$. 
If $K\subseteq E$, we say that $P(A,f)$ is terminating on $K$ if and only if
$P(A,f)$ terminates on every input $x\in K$.
Further, program $P(A,f)$ is non-terminating ($NT$ for short) if and only if it is non-terminating on some input $x\in E$. 
\end{df}
Thus, a program $P(A,f)\in P^{\H}$ is non-terminating if there is an input $x\in E$ such that $f(A^{k}(x)) > 0$ for all $k\geq 0$. 
We denote by $NT(P(A,f))$ the set of inputs $x\in E$ for which $P(A,f)$ is non-terminating.

Next, we introduce the important notion of an asymptotically non-terminating value \cite{scss2013rebiha}. 

\begin{df}\label{ant}
We say that $x\in E$ is an asymptotically non-terminating value for 
$P(A,f)$ if there exists some $k_x\geq 0$ such that $P(A,f)$ is non-terminating on 
$A^{k_x}(x)$. 
In this case, we will also say that $P(A,f)$ is $ANT$ on $x$, or that $x$ is $ANT$ for $P(A,f)$.
If $K\subseteq E$ we say that $P(A,f)$ is $ANT$ on $K$ if it is $ANT$ on every $x\in K$.  
We will also say that $P(A,f)$ is $ANT$ if it is $ANT$ for some input $x\in E$.
\end{df}
We denote by $ANT(P(A,f))$ the set of inputs $x\in E$ that are $ANT$ for $P(A,f)$.
The $ANT$ set has a central role in the study  and analysis of termination of program on any $A$-stable subset $K$ of $E$, as we will show.

\begin{ex} Consider again Example \ref{ex-motivation}.
We first note that the program terminates on $u=(-9,3,-2)^\top$ because with this initial value no loop iteration will be performed as $fA^0 u=-13/2$.
It is also easy to check that $fA^1u=-5/2$, and that $fA^2u=17.5$.
In fact we have  $fA^k u>0$ for all $k\geq 2$, so that 
the program is non-terminating on $A^2u = (63,3,22)^{\top}$.
We conclude that the initial value $u=(-9,3,-2)^{\top}$ belongs to the $ANT$ set.  
\end{ex}

The following theorem already shows the importance of $ANT$ sets:  termination for linear programs is reduced to the emptiness check of the $ANT$ set.

\begin{thm}\label{equivalence}
The program $P(\A,\f)$ in $P^{\H}$ is $NT$ if and only if it is $ANT$. 
More generally, if $K$ is an $\A$-stable subset of $E$, the program $P(A,f)$ is terminating on $K$  if and only if $ANT(P(A,f))\cap K$ is empty.
\end{thm}

\begin{proof}
It is clear that if $P(A,f)$ is $NT$, it is $ANT$ as a $NT$ value of $P(A,f)$ is, of course, also an $ANT$ value. 
Conversely, if $P(A,f)$ is $ANT$, let $x$ be an $ANT$ value.
Then $A^{k_x}(x)$ is a $NT$ value of $P(A,f)$, and so $P(A,f)$ is $NT$. The assertion for $A$-stable subspaces of $E$ is 
obvious, the proof being the same, as if $x\in K$ is $ANT$, we have $A^{k_x}(x)\in K$.
\end{proof}

The set of $NT$ values is included in the  $ANT$ set, but the most important property of an $ANT$ set resides in the fact that each of its elements gives an associated element in $NT$ for the corresponding program. That is, each element $x$ in the $ANT$ set, even if it does not necessarily belong to the $NT$ set, refers directly to initial values $A^{k_x}(x)$ for which the program does not terminate. 
Hence there is a number $k_x$ of loop iterations, departing from the initial value $x$, such that $P(A, f)$ does not terminate on $A^{k_x}(x)$. This does not imply that $x$ is $NT$ for $P(A,f)$ because the program $P(A,f)$ could terminate on $x$ by performing a number of loop iterations strictly smaller than $k_x$. 
On the other hand, the $ANT$ set is more than an over-approximation of the $NT$ set, as it provide us with a deterministic and efficient way to decide termination. 

Let ${ANT}^c$ be the complement of the $ANT$ set. It gives us an under approximation for the set of all initial values for which the program terminates. 

\begin{cor}\label{co-ANT}
Let $P(A,f)$ be  in $P^{\H}$. Then $P(A,f)$ terminates on the complementary set $ANT^c(P(A,f))$.   
\end{cor}

\begin{proof}
As $NT(P(A,f)) \subseteq ANT(P(A,f))$, passing to complementary sets gives the result. 
\end{proof}    

Theorem \ref{equivalence} provide a \emph{necessary and sufficient conditions} for the termination of linear programs. 
Further, it allows for the reduction of the problem of termination for linear programs to the emptiness check of the corresponding $ANT$ set.  
Also, Corollary \ref{co-ANT} shows that $ANT$ sets allow for the generation of initial
variable values  for which the loop program terminates.
In the following section we prove that the $ANT$ set is a semi-linear space,  and we show how it can be exactly and symbolically computed.

\section{Computation of ANT Sets for Homogeneous Programs}\label{1loop}
Let $P(A,f)$ be a program in $P^{\H}$.  In this section we start 
with Assumption ($\mathcal{H}$) bellow,
which will enable us to compute the sets $ANT(P(A,f))$ explicitly, and
will also help us show that such sets are semi-linear subspaces of $E$. 
In Section \ref{general} we show that a more general assumption reduces to this particular one, and in Section \ref{measure} we will show  
that this assumption is almost always satisfied, except for an extremely small class of programs.

\begin{ass}[$\mathcal{H}$]\label{a}
In this section we will assume that $Spec_\R(A)=Spec_{>0}(A)\cup \{0\}$, and that if $t$ 
is a positive eigenvalue of $A$ then no other eigenvalue of $A$ has the same module. 
\end{ass}

We denote by $\U$ the set of complex numbers of module $1$, \emph{i.e.}, $\U=\{z\in \C\ |\ |z|=1\}$.  We will need the following lemma.

\begin{LM}\label{oscillation}
Let $u=(u_1,\dots,u_r)$ be an element of $\U^r$, where 
$u_i\neq u_j$ and $u_i\neq \overline{u_j}$ when
$i\neq j$, $1\leq i,j\leq r$. 
Let $s_k$, $k\geq 0$, be as
$$s_k=\sum_{i=1}^n (a_i u_i^k+\overline{a_i}\overline{u_i}^k).$$
Then, either all 
$a_i$'s are zero, or there is a $c>0$ and there is an infinite number of $k$'s such that $s_k<-c$.
\end{LM} 
\begin{proof}
According to Lemma 4 of \cite{Braverman}, we know that either $s_k$ is constantly zero, or we are in the second situation of the statement. 
But if $s_k$ is constantly zero, then by 
Dedekind's theorem on linear independence of characters applied to $\Z$, all the $a_i$'s are zero.
\end{proof}

If $\tau$ is a positive real number in $|Spec(A)|-Spec_{>0}(A)$, we set $$A_\tau=\{\mu \in Spec_H(A),|\mu|=\tau\},$$ and 
$$\Sigma_\tau=\oplus_{\mu \in A_\tau} E_{\mu,\overline{\mu}}.$$ We also set $$U_\tau=\{u_\mu=\mu/\tau,\mu\in A_\tau\}.$$
The following proposition
is as a consequence of the Jordan basis Theorems \ref{Jnr} and \ref{JU}, in Section \ref{algebra}. 

\begin{prop}\label{partialsum}
If $\tau$ is a positive real number in $|Spec(A)|-Spec_{>0}(A)$, then for $x_\tau$ in $\Sigma_\tau$, the quantity $f(A^k (x_\tau))$ is of the form 
$$[\sum_{j=0}^{d_\tau-1} (\sum_{u_\mu \in U_\tau} a_{\mu,j}(x_\tau)u_\mu^k+\overline{a_{\mu,j}(x_\tau)u_\mu^k} )k^j]\tau^k,$$ where 
$d_\tau$ is the maximum of the integers $dim_\R(E_{\mu,\overline{\mu}})/2$, for $\mu\in A_\tau$, and the $a_{\mu,j}$'s are $\R$-linear maps 
from $\Sigma_\tau$ to $\C$, which can be computed explicitly.
If $t$ is a positive eigenvalue of $A$, for $x_t$ in $E_t$, the quantity $f(A^k (x_t))$ is of the form  
$$(\sum_{i=0}^{d_t-1} \alpha_{t,i}(x_t)k^i)t^k,$$ where $d_t$ is the dimension of $E_t$, and the $\alpha_{t,i}$'s are $\R$-linear maps 
from $E_t$ to $\R$, which can be computed explicitly.
\end{prop}


We are now going to describe the $ANT(P(A,f))$ sets as semi-linear spaces of $E$.
We note that the linear maps 
$\alpha_{t,i}$ and $a_{\mu,j}$, in Proposition~\ref{partialsum}, can be computed easily. 
In our previous  work \cite{TR-IC-14-09},  we showed how these linear maps are computed efficiently
for programs over the reals and for programs over $\Z$, 
when the induced matrix $A$ had a real spectrum.  
Here, the computation of $\alpha_{t,i}$ and $a_{\mu,j}$ remains similar to those described in \cite{TR-IC-14-09}, Sections $7$ and $8$.

We first introduce the following subsets of $E$.

\begin{df}
For $t\in Spec_{>0}(A)$, and $l$ between $0$ and $d_t-1$, we define the sets $S_{t,l}$ to be the sets of elements 
$x$ in $E$ which satisfy:
\begin{itemize}
\item For $\tau>t$ in $|Spec(A)|$, 
\begin{itemize}
\item if $\tau \notin Spec(A)$, then for all $\mu \in A_\tau$ and $j\in\{0,\dots,d_\tau-1\}$: \begin{equation}a_{\mu,j}(x_\mu)=0.\end{equation}\label{eq1} 
\item if $\tau \in Spec(A)$, then for all $i\in\{0,\dots,d_\tau-1\}$ \begin{equation}\alpha_{\tau,i}(x_\tau)=0.\end{equation}\label{eq2} 
\end{itemize}
\item For all $i$ between $l+1$ and $d_t-1$,  \begin{equation}\alpha_{t,i}(x_t)=0.\end{equation}\label{eq3} 
\item Finally we have the inequalities: 
\begin{equation}\alpha_{t,l}(x_t)>0.\end{equation}\label{eq4} 
\end{itemize}
\vspace*{-6ex}
\end{df}

We can now  state the main result of this section, describing the generic formulas representing exactly and symbolically the $ANT$ sets.

\begin{thm}\label{generation-ANT}
The set $ANT(P(A,f))$ is the disjoint union of the sets $S_{t,l}$, for $t\in Spec_{>0}(A)$, and $l\in \{0,\dots,d_t-1\}$. In other words, considering the set $\Delta_{S}=\{(t,l)\ |\  t\in Spec_{>0}(A), l\in \{0,\dots,d_t-1\}\}$ we have $$ANT(P(A,f))= \bigvee_{(t,l)\in\Delta_S} S_{t,l}.$$ 
\end{thm}

\begin{proof}
First, if $x$ belongs to $S_{t,l}$ then, by assumption, the sequence $fA^kx$ will be asymptotically equivalent to $$t^k\alpha_{t,l}(x_t)k^l,$$ 
which grows without bound.
Hence, $x$ belongs to $ANT(P(A,f))$.

 Conversely, suppose that $x$ belongs to none of the $S_{t,l}$ sets. Let $\tau$ be the highest absolute value among the eigenvalues of $A$, 
such that for $\tau'>\tau$,  $\tau'$ being the module of an eigenvalue of $A$, we have $$a_{\mu,j}(x_\mu)=0$$ for all $\mu \in A_{\tau'}$ and $j\in\{0,\dots,d_{\tau'}-1\}$ when 
$\tau'\notin Spec(A)$, and $$\alpha_{\tau',i}(x_{\tau'})=0$$ for all $i\in\{0,\dots,d_{\tau'}-1\}$ when $\tau'\in Spec(A)$.
Then, 
\begin{itemize}
\item If $\tau=0$, we get $f(A^k(x))=f(A^k(x_0))$, which is constantly zero for $k$ large enough.
Hence, $x$ is not in $ANT(P(A,f))$.\\ 
\item If $\tau>0$, we have two possibilities, depending on whether  
$\tau$ is  in $Spec(A)$, or not.
\begin{itemize}
\item  When $\tau\not\in Spec(A)$, let $l$ be the highest integer between $0$ and $d_\tau-1$ such that $a_{\tau,l}(x_\tau)$ is nonzero. 
 We know, from  Lemma \ref{oscillation}, that for an infinite number of 
$k$'s, the sum $$s_k=\sum_{u_\mu \in U_\tau} a_{\mu,l}(x_\tau)u_\mu^k+\overline{a_{\mu,l}(x_\tau)u_\mu^k}$$ is smaller than a negative number $-c$, which is independent of $k$. 
As the integers $k$ grow it follows, from Proposition \ref{partialsum}, that $fA^kx$ will be equivalent to $fA^kx_\tau$.
By the choice of $l$, the latter itself be equivalent 
to $$s_kk^l\tau^k\leq -ck^l\tau^k,$$ 
which decreases without bound.
Hence, $fA^kx$ will be negative for an infinite number 
of $k$'s, and thus $x$ does not belong to $ANT(P(A,f))$.
\item Nw assume $\tau=t\in Spec_{>0}(A)$, and let $l$ be the highest integer between $0$ and $d_t-1$ such that $\alpha_{t,l}\neq 0$.
Then, as $x$ is not in $S_{t,l}$, we must have $$\alpha_{t,l}(x_t)<0.$$ But $fA^kx$ is equivalent to $$\alpha_{t,l}(x_t)t^k$$ when $k$ grows according to Proposition \ref{partialsum}.
Hence, $fA^kx$ decreases without bound, and so $x$ is not in $ANT(P(A,f))$. 
This completes the proof.
\end{itemize}
\end{itemize}
\end{proof}

In the next section we generalize these results to programs in the  classes $P^{\G}$ and $P^{\A}$.  
We  show that the problem of generating  $ANT$ sets for linear and affine programs reduces to the computation  $ANT$ sets for  specific homogeneous programs under Assumption ($\mathcal{H}$).

\section{Termination over $\Z$ for Linear and Affine Programs}\label{general}
In this section, we extend our methods to  linear  and affine  programs. 
For each of these program classes, the $ANT$ set generation problem is reduced to the computation of $ANT$ sets of corresponding  homogeneous programs under Assumption ($\mathcal{H}$). 

For $f_1,\dots,f_r$ a family of elements in $E^*$, $b\in\R^r$, and $c$ a vector of $E$, we consider the \textit{affine program} 
$P(A,F,b,c)=P(A,(f_i)_{i=1,\dots,r},b,c)\in P^{\A}$: 
$$P(A,(f_1,\dots,f_r),b,c): \texttt{while ($\bigwedge_{1\leq i\leq r}$fi(x)>bi) \{x:=Ax+c\}}.$$ 
We will also consider the \textit{linear program} $P(A,F)=P(A,(f_i)_{i=1,\dots,r})$, where
$$P(A,F)=P(A,(f_1,\dots,f_r))=P(A,(f_1,\dots,f_r),0,0).$$

\subsection{ANT sets for generalized homogeneous programs}\label{ghp}

First , we remove some restrictions on $A$.
We denote by  $R(A)$ the set of nonzero eigenvalues of $A$ with arguments a rational multiple of $2\pi$, \emph{i.e.}, 
$$R(A)=\{\l\in Spec(A),\ Arg(\l)\in 2\pi\Q\}=\{\l\in Spec(A), \ \exists n \in \N, \ \lambda^n>0 \}.$$ 

\begin{ass}[$\mathcal{G}$]\label{a1}
For any eigenvalue $\l$ of $A$ in $R(A)$, if $\mu \in Spec(A)-\{0\}$ is such that $|\mu|=|\l|$, then $\mu$ is equal to $\l$ up to a root of unity in $\C$, i.e. 
if $\l$ and $\mu$ in $Spec(A)-\{0\}$ are such that $|\l|=|\mu|$, then either both are in $R(A)$, or none is. In other words, 
two nonzero eigenvalues with the same module both have an argument which is either a rational multiple of $2\pi$, or none has.
\end{ass}

From now on, we suppose that $A$ satisfies Assumption ($\mathcal{G}$).
As the rational numbers are a negligible set of $\R$, we  see that for a generic matrix $A$, the set $R(A)$ is empty.
Hence, 
almost all matrices $A$ in $\mathcal{M}(n,\R)$ satisfy Assumption ($\mathcal{G}$). In Section \ref{measure}, we will confirm this fact, and actually show more precisely 
that the set of matrices satisfying Assumption ($\mathcal{G}$) contains a dense open set of total measure, \emph{i.e.}, whose complement is of measure zero.


First, we show that we can reduce the computation of the $ANT$ set for an homogeneous program $P(A,f)\in P^{\H}$, when $A$ satisfies Assumption ($\mathcal{G}$), to the intersection of 
the $ANT$ sets of programs $P(G,g)\in P^{\H}$, with $G$ satisfying Assumption ($\mathcal{H}$). This reduction technique is also used in \cite{JO}. 
First, we notice that an appropriate power of $A$ satisfies Assumption ($\mathcal{H}$). 

\begin{prop}\label{AN}
Let $Q$ be the set defined by $$Q=\{\mu/|\mu|, \ \mu \in R(A) \}.$$ 
If $N$ is the \emph{lcm} of the orders of elements of $Q$, then $A^{N}$ satisfies Assumption ($\mathcal{H}$).
\end{prop}
\begin{proof}
If $r$ belongs to $Spec_{>0}(A^N)$, let $\lambda$ be an eigenvalue of $A^N$, such that $|\l|=r$. 
As $\l$ is in $Spec(A^N)$, it is equal to $\mu^N$ for some $\mu \in Spec(A)$, 
which is in fact in $R(A)$ as $\mu^N>0$.
Similarly, $r=\mu'^N$ for $\mu'\in Spec(A)$. 
But $\mu'^N=r>0$, and so $\mu'\in R(A)$. 
As $$|\mu|=|\mu'|=r^{1/N},$$ by  Assumption ($\mathcal{G}$), $\mu'$ is also 
in $R(A)$.  But then, by the
definition of $N$, we have $\mu'^N>0$, and so $\l=\mu'^N=r$. 
Thus, the second part of Assumption $(\mathcal{H})$ is satisfied. 

Moreover, if $A^N$ had a negative eigenvalue $\l$, again it would be of the form $\l=\mu^N$. But then we would have $\mu^{2N}=\l^2>0$ and so, 
by the definition of $N$, we would get $(\mu/|\mu|)^N=1$.
That is, $\l=|\mu|^N$, which is absurd. 
Hence, $A^N$ also satisfies the first part of Assumption $(\mathcal{H})$.
\end{proof}

We recall that, in the previous section, we showed that if $G$ satisfies assumption ($\mathcal{H}$), then for any $g\in E^*$ the set $ANT(P(G,g))$ is semi-linear, and we computed it explicitly. Now, we show how to compute $ANT(P(A,f))$.

\begin{thm}\label{H-ANT}
We have $ANT(P(A,f))=\cap_{l=0}^{N-1} ANT(P(A^N,fA^l))$.
\end{thm}
\begin{proof}
It is clear that $$ANT(P(A,f))\subset\cap_{l=0}^{N-1} ANT(P(A^N,fA^l)).$$ Conversely, if $x$ belongs to $\cap_{l=0}^{N-1} ANT(P(A^N,fA^l))$. Then for 
every $l$ between $0$ and $N-1$, there is $m_{x,l}$, such that $k\geq m_{x,l}$, which gives $f A^{kN+l}x>0$. Taking $$m_x=max_{l\in\{0,\dots,N-1\}} m_{x,l},$$ 
we have $$k\geq m_x \Rightarrow f (A^k (x))>0,$$ that is, $x\in ANT(P(A,f))$. This proves the equality.
\end{proof}

Proposition \ref{AN} guarantees that the matrix $A^N$ satisfies Assumption ($\mathcal{H}$). Considering a program $P(A,f)\in P^{\H}$, with $A$ satisfying Assumption ($\mathcal{G}$), Theorem \ref{H-ANT} shows that the $ANT(P(A,f))$ set
is the intersection of the 
$ANT(P(A^N,fA^l))$ sets, with $A^N$ satisfying Assumption ($\mathcal{H}$).
This handles the case of linear homogeneous programs  with one loop condition under assumption ($\mathcal{G}$). 
Now, as in \cite{TR-IC-14-09}, we reduce the computation of the $ANT$ set of a generalized homogeneous program  to that of a homogeneous program.

Consider $P(A,F)=P(A,(f_i)_{i=1,\dots,r})$ in $P^{\G}$.
We start with the following lemma on non-terminating values. 

\begin{df}
The value $x$ is $NT$ for $P(A,F)$ in $P^{\G}$ if and only if it is $NT$ for all $P(A,f_i)$, with $i\in \{1,\dots,r\}$.
\end{df}

Now, we define $ANT$ values  for such programs.

\begin{df}\label{antG}
We say that $x$ is $ANT$ for $P(A,F)$ if there exists $k_x$ such that for all $i\in\{1,\dots , r\}$ we have $f_i(A^k(x)) > 0$ for 
$k>k_x$, that is, if $x$ is $ANT$ for all programs $P(A,f_i)$.
\end{df}  

Again we have the following easily proved but important lemma.

\begin{LM}\label{ANTandNT}
A program $P(A,F)$ is $NT$ if and only if it is $ANT$, that is, $ANT(P(A,F))\neq \emptyset$.
\end{LM}
\begin{proof}
If $x$ belongs to $NT(P(A,F))$, then it belongs to $ANT(P(A,F))$. Conversely, if $x$ belongs to $ANT(P(A,F))$, then 
for some $k$, by definition, $A^k(x)$ belongs to $NT(P(A,F))$. In particular, 
both sets are empty or non empty together, which proves the claim. 
\end{proof}

We can now express the $ANT$ set of programs in $P^{\G}$ as the intersection of $ANT$ sets from corresponding programs in $P^{\H}$. 

\begin{prop}\label{gen2hom}
Let $f_1,\dots,f_r$ be linear forms on $E$, then one has 
$$ANT(P(A,(f_1,\dots,f_r)))=\cap_{i=1}^r ANT(P(A,f_i)).$$
\end{prop} 
\begin{proof}
If $x$ is in $ANT(P(A,(f_1,\dots,f_r)))$, there is $k\geq 0$ such that $f_i A^l(x)>0$ for $l\geq k$, for all $i$, hence 
$x$ belongs to every set $ANT(P(A,f_i))$. Conversely, if $x$ belongs to $\cap_{i=1}^r ANT(P(A,f_i))$, then 
for each $i$, there is $k_i\geq 0$, such that $l\geq k_i$ implies $f_i A^l(x)>0$. Take $k=max_i(k_i)$, then 
$l\geq k$ imples that for every $i$, $f_i A^l(x)>0$, i.e. $x$ belongs to $ANT(P(A,(f_1,\dots,f_r)))$.
\end{proof}

\subsection{ANT sets for affine programs over $\Z$}\label{ap}
First, we define the notion of $ANT$ values for affine programs.  
\begin{df}
Let $P(A,F,b,c)$ be an affine program in
$P^{\mathbb{A}}$. For $x=x_0\in \R^n$, denote by $x_1$ the vector $Ax+ c$ and, 
recursively, let  
$x_k=Ax_{k-1}+c$, $k\geq 1$. 
We say that a vector $x$ is $ANT$ for $P(A,F,b,c)$ if there is some $k_x$ such that $k\geq k_x$ implies $Fx_k>b$. 
We denote by $ANT(P(A,F,b,c))$ the set of $ANT$ input values of $P(A,F,b,c)$.
\end{df}

Consider  the affine program $P(A,F,b,c)=P(A,(f_1,\dots,f_r),b,c)$.
We denote by $E'$ the vector space $E\oplus \R$. 
We denote by 
$A'$ the linear map from $E'$ to itself defined by $$A':x+t\mapsto (Ax+tc)+t,$$ and 
let $f'_i$, $1\leq i\leq r$, be the linear form on $E'$ 
defined by $$f'_i:x+t\mapsto f_i(x)-tb_i,$$ and let $f'_{r+1}:x+t\mapsto t$. Finally, for $x$ in $E$, we set $x'=x+1$ in $E'$. As we have $$Spec(A')=Spec(A)\cup \{1\},$$ we notice at once the following fact.

\begin{fact}\label{a2a} 
$A'$ satisfies Assumption ($\mathcal{G}$), if and only if $A$ satisfies it, and no eigenvalue in $Spec(A)-R(A)$ has module $1$.
\end{fact}

We make this conclusion explicit.
\begin{ass}($\mathcal{A}$)\label{a2}
Let $P(A,F,b,c)$ be an affine program in
$P^{\mathbb{A}}$. 
We say that $A$ satisfies Assumption ($\mathcal{A}$) when it satisfies Assumption ($\mathcal{G}$) and no eigenvalue in $Spec(A)-R(A)$ has module $1$.
\end{ass}

When working with an affine program $P(A,F,b,c)$, also written as $P(A,(f_1,\dots,f_r),b,c)$, we will assume that $A$ satisfies Assumption ($\mathcal{A}$). 

\begin{prop}\label{aff2gen}
The input $x$ is in the set $ANT(P(A,(f_1,\dots,f_r),b,c))$ 
if and only if input $x'$ is in the set $ANT(P(A',(f_1',\dots,f_r')))$.
\end{prop}

\begin{proof}
 Fix $B$ a basis of $E$, and let $\Ab=Mat_B(A) \in \mathcal{M}(n,\R)$, $\F\in \mathcal{M}(r,n,\R)$ the matrix 
with rows equal to the $Mat_B(f_i)$'s, $\bb= Mat_B(b)$  in $\mathcal{M}(1,r,\R)$, and let $\cb=Mat_B(c)$. Let 
$B'$ be the basis of $E\oplus \R$, with first vectors $(e_i,0_\R)$, for $e_i$ in $B$, and last vector 
$(0_E,1)$. Now let $\Ab'=Mat_{B'}(A')\in \mathcal{M}(n+1,\R)$ and $\F'\in
\mathcal{M}(r+1,n+1,\R)$ the matrix with rows $Mat_{B'}(f'_i)$.
Clearly, we have 
$\Ab'=\begin{bmatrix} \Ab &\cb \\ 0 & 1\end{bmatrix}$, and
$\F'=\begin{bmatrix} \F &-\bb \\ 0 & 1\end{bmatrix}$.   
To say that $(x, 1)^{\top}$ is $ANT$ for $P(\Ab',\F')$ means that there exists $k_x$, such that when $k\geq k_x$, we get $\F' \Ab'^k \cdot (x,1)^{\top}>0$. We define $x_k$ by induction, as $x_0=x$, and $x_{k+1}=Ax_k+c$.
But as 
$\Ab' \cdot (x, 1)^{\top}= \begin{small}\begin{pmatrix} \Ab x+\cb \\ 1 \end{pmatrix}\end{small} = (x_1 , 1)^{\top}$, by induction, we obtain $\Ab'^k \cdot (x , 1)^{\top} = (x_k , 1)^{\top}$, $k\geq 1$. 
But
$\F' \Ab'^k \cdot (x ,1 )^{\top}=\F' \cdot (x_k, 1)^{\top} =  \begin{small}\begin{pmatrix} \F x_k-\bb \\ 1 \end{pmatrix}\end{small}$. Hence, 
$\F' \Ab'^k \cdot (x , 1)^{\top}>0$ is equivalent to $\F x_k> \bb$, and the result follows.
\end{proof}

Proposition \ref{aff2gen} shows that the generation of the $ANT$ set for a program in $P^{\mathbb{A}}$ reduces to the generation of the $ANT$ set for an associated program in $P^{\G}$, and that 
reduces to the computation of $ANT$ sets for corresponding homogeneous programs in $P^{\H}$.
These two reduction provide us with computational methods for 
the automatic generation of $ANT$ sets for affine programs under Assumption ($\mathcal{A}$). Now, we can state the following termination result for generalized homogeneous and affine programs over $\Z$. 

\begin{thm}\label{decid}
Let $A\in \mathcal{M}(n,\Z)$ be a matrix over the integers associated to  loop instructions. 
\begin{itemize}
\item Then $P(A,(f_1,\dots,f_r))$ terminates on $\Z^n$ if and only if 
$$ANT(P(A,f_1,\dots,f_r))\cap \Z^n=\emptyset,$$ 
\item $P(A,(f_1,\dots,f_r),b,c)$ terminates on $\Z^n$ if and only if 
$$ANT(P(A',(f_1',\dots,f_r')))\cap \Z^n\times 1=\emptyset.$$ 
\end{itemize}
\end{thm}

\begin{proof}
We know that the computation of the $ANT(P(A,(f_1,\dots,f_r)))$ sets reduces to the  intersection of the $ANT(P(A,f_i))$ sets, with $1\leq i \leq r$,
by Proposition \ref{gen2hom}. 
For the $ANT(P(A,f_i))$ sets we  apply Theorem \ref{equivalence}, with  
$K=\Z^n$, and thus establish the first assertion.
We saw that $x\in K\subseteq E$ is $ANT$ (resp. NT) for $P(A,F,b,c)$ if and only if $x'=(x, 1)^{\top}$ is $ANT$ (resp. NT) for $P(A',F')$, with 
$A'=\begin{bmatrix} A & c \\ & 1 \end{bmatrix}$, and $F'=\begin{bmatrix} F &-b \\ 0 & 1\end{bmatrix}$. We  apply Theorem \ref{equivalence} with
$K'=\{x',x\in \Z^n\}$, which is $A'$-stable.
\end{proof}

The following corollary states the main decidability result for the termination problem for affine programs over the integers.

\begin{cor}\label{z-decidability}
Under Assumption $(\mathcal{G})$, the termination over $\Z^n$ of programs in the form $P(A,f_1,\dots,f_r)$ is decidable. 
Under Assumption ($\mathcal{A}$), the termination of programs $P(A,f_1,\dots,f_r,b,c)$ 
over $\Z^n$ is decidable.
\end{cor}

\begin{proof}
We appeal to a result in \cite{Khachiyan97}, which asserts that it can be decided 
if a convex semi-algebraic subspace of $\R^n$, contains an element 
of the lattice $\Z^n$. We apply it to the subspace $ANT(P(A,f_1,\dots,f_r))$ of $\R^n$ in the first case. In the second case, 
we apply it to the image of the subspace $$\R^n\times \{1\}\cap ANT(P(A',f_1',\dots,f_r'))$$ of $\R^n\times \{1\}$ under the canonical projection 
from  $\R^n\times \{1\}$ to $\R^n$.
\end{proof}

Corollary \ref{z-decidability} provides the most complete response to the termination problem left open in \cite{Braverman}.

\section{Matrices Satisfying Assumptions ($\mathcal{G}$) or ($\mathcal{A}$)}\label{measure}
In this section, we show that Assumptions ($\mathcal{G}$) or ($\mathcal{A}$) are almost always satisfied. 

\begin{thm}\label{thm-mes}
The set of matrices $A$ in $\mathcal{M}(n,\R)$ satisfying Assumption $(\mathcal{G})$ contains a dense open subset of $\mathcal{M}(n,\R)$, and of total Lebesgue measure in $\mathcal{M}(n,\R)$. The same assertion is true for matrices satisfying Assumption ($\mathcal{A}$).
\end{thm}
\begin{proof}
We consider the set $U$ of $\mathcal{M}(n,\R)$, of semi-simple ---~that is, diagonalizable over 
$\C$,~--- matrices with distinct eigenvalues. It is the complement set 
of zeros of $$P:A\mapsto disc(\chi_A,\chi_A'),$$ 
where $disc$ stands for the discriminant.
Thus, $disc$ is dense, open, and of total measure. We denote by $W$ the open subset of $\C^n$, consisting of 
$n$-uples $(z_1,\dots,z_n)$, which satisfy $z_i\neq z_j$ when $i\neq j$. 

Let $\sigma_i(z_1,\dots,z_n)$ denote the coefficient of $X^i$ in $(X-z_1)\dots(X-z_n)$. 
It is well known that the map $\sigma$ from $W$ to $\C^n[X]$, defined by 
$$\sigma:(z_1,\dots,z_n)\mapsto (\sigma_0(z_1,\dots,z_n),\dots,\sigma_{n-1}(z_1,\dots,z_n)),$$ 
is a local diffemorphism, as its Jacobian at $z$ equals,  up to the sign, the product $$\prod_{i<j}(z_i-z_j).$$ 
We are going to show the set of matrices in $U$, which do not have two non-conjugate eigenvalues with the same absolute value is open, dense, and of total measure 
in $\R^n$. As this set is contained in the set of matrices satisfying ($\mathcal{G}$), this will prove our first assertion.

Let  $\C_{n,1}[X]^{reg}$ denote the set of monic polynomials of $\C_n[X]$ with distinct roots, and let 
$P$ be such a polynomial.
Number its roots as $(z_1(P),\dots,z_n(P))\in W$.
Then, there is an open neighborhood $\mathcal{N_P}$ of $P$ in $\C_{n,1}[X]^{reg}$ such that, 
for $Q$ in $\mathcal{N_P}$, one can number the roots of $Q$ as  $(z_1(Q),\dots,z_n(Q))\in W$, and 
$$R:Q\mapsto (z_1(Q),\dots,z_n(Q))$$ is a smooth diffeomorphism from $\mathcal{N_P}$ to its open image $R(\mathcal{N_P})\subset W$. 
Hence, if $P$ belongs to $\R_{n,1}[X]^{reg}= \C_{n,1}[X]^{reg}\cap \R_n[X]$, we have that  $$\mathcal{N_P}^r=\mathcal{N_P}\cap \R_n[X]\subset \R_{n,1}[X]^{reg}$$ is a submanifold of $\mathcal{N_P}$, $R(\mathcal{N_P}^r)$ is a submanifold of 
$R(\mathcal{N_P})$, and the restriction of $R$ to $\mathcal{N_P}^r$ is thus a smooth diffeomorphism to its image $R(\mathcal{N_P}^r)$. In fact, it is easy to see what 
$R(\mathcal{N_P}^r)$ looks like.

We denote by $B(u,\epsilon)$ the open ball of radius $\epsilon>0$ around the complex number $u$. Suppose that $(z_1(P),\dots,z_n(P))$ is ordered in such a way that $z_1(P),\dots,z_a(P)$ are real,
and the other roots come in $b$ couples of conjugate complex numbers $(z_i(P),z_{i+1}(P))$ with $z_{i+1}(P)=\overline{z_i(P)}$ and  $n=a+2b$.
Then, one can choose 
$\mathcal{N_P}$ such that for some positive $\epsilon$ , $R(\mathcal{N_P}^r)$ is diffeomorphic to the product
\begin{align*}
]z_1(P)-\epsilon,z_1(P)+\epsilon[\,\,\,\times\, \dots\, \times\,\,\, ]z_a(P)-\epsilon,z_a(P)+\epsilon[\\ \times B(z_{a+1},\epsilon)\times B(z_{a+3},\epsilon)\dots \times B(z_{a+2b-1},\epsilon).\end{align*}
In particular, the intersection of $R(\mathcal{N_P}^r)$ with the set $|z_i|=|z_j|$ when $i$ and $j$
are such that 
$z_i(P)$ and $z_j(P)$ are not conjugate, is a hypersurface of $R(\mathcal{N_P}^r)$. 

Finally, as the map $$A\in U\mapsto \chi_A\in \R_{n,1}[X]^{reg}$$ is submersive everywhere 
then, the set of matrices in $U$, which have two distinct non conjugate eigenvalues with the same module, is locally the union of at most $n(n-1)/2$ hypersurfaces. In particular, 
its complementary set is open, dense, and of total measure in $U$, hence in $\R^n$. 
We have proved our first assertion. 

The second assertion's proof is completely similar.
\end{proof}

\section{Discussion}\label{discussion}

In  this section we note some related works.
Then we summarize some of our previous results 
along similar lines, and list the main contributions presented here.

\subsubsection{Related work\rm:}
Concerning the \emph{termination analysis} for affine programs over the reals, rationals and the integers, we reduced the problem to the emptiness check of the generated $ANT$ sets.
By so doing, we obtained a characterization of terminating linear programs which allows for a practical and computational procedure. 
In~\cite{Braverman, tiwaricav04}, the authors focused on the decidability of the termination problem for linear loop programs.  
Also, the techniques in \cite{Braverman} are based on the approach in \cite{tiwaricav04}, but now considering termination analysis over the rationals and integers for \emph{homogeneous programs} only. 
But the termination problem for \emph{general affine} programs over the integers is left open in~\cite{Braverman}. 

Recently, in \cite{JO}, considering the $ANT$ set and a technique similar to our approach previously proposed in \cite{scss2013rebiha, TR-IC-14-09}, the authors were able to answer this question for programs with semi-simple matrices, using strong results from analytic number theory, and diophantine geometry. 
By contrast, in \cite{JO} the author focus on decidability results, and the $ANT$ set is not explicitly computed there. In fact, the  $ANT$ set is referred to  as a semi-algebraic set and  the use of 
quantifier elimination techniques is suggested.  
  In this work, although we also considered the termination problem, 
we addressed a more general problem, namely, the \emph{conditional termination} problem of generating static sets of terminating and non-terminating inputs. 
We provide efficient computational methods allowing for the exact computation and symbolic representation of the $ANT$ sets for affine loop programs over $\R$, $\Q$,  $\Z$, and $\N$.
The $ANT$ sets generated by our approach can be seen as a precise over-approximation
for the set of non-terminating inputs.
We use ``precise'' in the sense that $NT \subseteq ANT$ and all elements in $ANT$, even those  not in $NT$, are directly associated with non-terminating values, modulo a finite numbers of loop iterations.  
The, possibly infinite, complement of an $ANT$ set is also a ``precise'' under-approximation of the set of terminating inputs, as it provides terminating input data entering the loop at least once.  

Our method differs from those proposed in \cite{Cook:2008}, as we do not use the synthesis of ranking functions.  

The methods proposed in \cite{Gulwani:2008} can provide non-linear preconditions, but we always generate semi-linear sets as precondition for termination, which facilitates the static analysis of liveness properties. 

The approach in \cite{Bozga} considers first octagonal relations and the associated class of formulae representing weakest recursive sets. 
It also suggests the use of quantifier elimination techniques and algorithms, which would 
require an exponential running time complexity of order $O(n^3\cdot 5^n)$, where $n$ is the number of variables. 
They also consider the conditional termination problem for restricted subclasses of
linear affine relations, where the  associated matrix has to be diagonalizable
and  with all non-zero eigenvalues of multiplicity one. 
They also identify other classes where the generated precondition would be non-linear.   

The experiments in \cite{samir:2013}, involving handwritten programs, are handled successfully by our algorithm presented in a companion article \cite{TR-IC-14-09}, more oriented towards static program analysis. The strength and the practical efficiency of the approach is shown by our
experiments dealing with a large number of larger linear loops. 
In \cite{TR-IC-14-09}, we present several details related to the application of the theoretical contributions exposed here. 
Our prototype was tested and the average time to generate the $ANT$ over $9000$ randomly generated loops was 0.75 seconds. In this experiment, the associated  matrices were triangularizable, with a number of variables between $3$ and $15$, and a number of conjunctions forming the loop condition between $1$ to $4$. 
In this more static program analysis applied work, we used examples from \cite{samir:2013,Cook:2008,Bozga,Braverman,2013:POPL,Podelski04}. 

\subsubsection{Our prior work\rm:}
We list here the points most relevant to the present discussion. 
\begin{itemize}
\item In~\cite{TR-IC-13-07, TR-IC-13-08} we provided new termination analysis algorithms that ran in polynomial time complexity. 
\item We considered the set of \emph{asymptotically non-terminating initial variable values} for the first time in ~\cite{scss2013rebiha}.
In that work we generated the $ANT$ set for a restricted class of linear programs over the reals, with only one loop condition, and where the associated linear forms of the loop lead to diagonalizable systems.
\item  In \cite{TR-IC-14-09} we showed how to automatically generate the $ANT$ sets for linear and affine programs. In that work, we also handled the case of linear or affine programs over $\Z$ with transition matrices admitting a real spectrum. It is the first substantial contribution on termination of linear program over the integers.
Here, we removed these restrictions.  
In \cite{TR-IC-14-09}, we also treated the case of matrices with a real spectrum. But if that is the case,
if two distinct eigenvalues have the same module, 
one is the opposite of the other, that is, they are equal up to the root of unity. 
In particular, in this case, Assumption 
($\mathcal{G}$) is always satisfied,  and so 
the results obtained here fully generalize those obtained in \cite{TR-IC-14-09}.

\end{itemize}

\subsubsection{The main contributions\rm:}
The central contributions presented in this article are listed below.
\begin{itemize}
\item Our criteria for termination over stable subspaces allowed us to show that termination for linear or affine programs over $\Z$ is decidable for almost the whole class of such programs.
\item We proved that the $ANT$ set is a semi-linear space, and we provided a computational method allowing for their automatic generation. 
\item We rigorously proved  that our assumption holds for almost all linear or affine programs by showing that the excluded programs forms an extremely small set of zero Lebesgue measure.
\item Our main results, Theorems  \ref{equivalence}, \ref{generation-ANT}, \ref{H-ANT}, \ref{gen2hom}, \ref{aff2gen}, \ref{decid}, \ref{z-decidability}, and \ref{thm-mes}, are evidences of the novelty of our approach. 
\end{itemize}

\section{Conclusions}\label{conclusion}
In terms of decidability results, we provide the most complete response to the termination problem 
for linear or affine programs over the integers.
We reduced the termination problem of linear, affine programs over $\Z$ to the emptiness check of the $ANT$ set of corresponding homogeneous linear programs. 
Then, we proved that these sets are semi-linear spaces which are easy to compute and manipulate.  

These theoretical contributions are mathematical in nature with proofs that are quite technical. We showed, however, that these results 
can be directly applied in practical ways.
One can rely the ready-to-use formulas representing the $ANT$ set provided in this article. 

Also, any static program analysis technique could incorporate, by a simple and direct instantiation, 
the generic ready-to-use formulas representing the preconditions for termination
and non-termination.  

\bibliographystyle{splncs}
\bibliography{termination}

\end{document}